\documentclass[final, onecolumn, 10pt]{IEEEtran}
\usepackage[ruled,vlined]{algorithm2e}
\usepackage{amsmath,epsfig, citesort,cite}
\usepackage{algorithmicx}
\usepackage{algpseudocode}
\usepackage{url}
\usepackage{amssymb,amsthm}
\usepackage{times,amsmath,amssymb,amsfonts,epsfig,graphicx}
\usepackage{enumerate}
\usepackage{stackrel}
\usepackage{multirow}
\usepackage{epsfig}
\usepackage{subfigure}
\usepackage{psfrag}
\usepackage{color}
\usepackage{comment}
\usepackage{xspace}

\newcommand{\clash}{{\textsc{Clash}}\xspace}

\newcommand{\bigO}{\mathcal O}

\newcommand{\abs}[1]{\left\vert#1\right\vert}

\newcommand{\vectornorm}[1]{\|#1\|}
\newcommand{\vectornormbig}[1]{\Big\|#1\Big\|}
\newcommand{\vectornormmed}[1]{\big\|#1\big\|}

\newcommand{\obs}{y}
\newcommand{\sensing}{\Phi}
\newcommand{\signal}{x}
\newcommand{\bestsignal}{x^\ast}
\newcommand{\noise}{\varepsilon}
\newcommand{\dimension}{n}
\newcommand{\numsam}{m}
\newcommand{\sparsity}{k}

\newcommand{\constraint}{\mathcal{C}_\sparsity}
\newcommand{\constrainttwo}{\mathcal{C}_{2\sparsity}}

\DeclareMathOperator*{\argmax}{arg\,max}
\DeclareMathOperator*{\argmin}{arg\,min}

\newcommand{\R}{ \mathbb{R} } 

\newtheorem{lemma}{Lemma}
\newtheorem{prop}{Proposition}
\newtheorem{definition}{Definition}
\newtheorem{remark}{Remark}
\newtheorem{theorem}{Theorem}
\newtheorem{corollary}{Corollary} 

\begin{document}

\title{Combinatorial Selection and Least Absolute Shrinkage via the \clash Algorithm}

\IEEEoverridecommandlockouts

\author{Anastasios Kyrillidis and Volkan Cevher\thanks{This work was supported in part by the European Commission under Grant MIRG-268398, ERC Future Proof, and DARPA KeCoM program
\#11-DARPA-1055. VC also would like to acknowledge Rice University for his Faculty Fellowship.}\\
Laboratory for Information and Inference Systems,
\'{E}cole Polytechnique F\'{e}d\'{e}rale de Lausanne}

\maketitle
\begin{abstract}
The least absolute shrinkage and selection operator (LASSO) for linear regression exploits the geometric interplay of the $\ell_2$-data error objective and the $\ell_1$-norm constraint to arbitrarily select sparse models. Guiding this uninformed selection process with sparsity models has been precisely the center of attention over the last decade in order to improve learning performance. To this end, we alter the selection process of LASSO to explicitly leverage combinatorial sparsity models (CSMs) via the combinatorial selection and least absolute shrinkage (\clash) operator. We provide concrete guidelines how to leverage combinatorial constraints within \clash, and characterize \clash's guarantees as a function of the set restricted isometry constants of the sensing matrix. Finally, our experimental results show that \clash can outperform both LASSO and model-based compressive sensing in sparse estimation.
\end{abstract} 

\IEEEpeerreviewmaketitle

\section{Introduction}\label{sec: introduction}
The {\em least absolute shrinkage and selection operator} (LASSO) is the {\it de facto} standard algorithm for regression \cite{tibshirani96regression}. LASSO estimates sparse linear models by minimizing the empirical data error via: %\vspace{-0.2cm}
\begin{equation}\label{eq: LASSO}
\widehat{\signal}_{\text{LASSO}} = \arg\min \left\{\left\|\obs - \sensing \signal \right\|_2^2:~~\|\signal\|_1\le \lambda \right \}, %\vspace{-0.2cm}
\end{equation}
where $\|\cdot\|_r$ is the $\ell_r$-norm. In \eqref{eq: LASSO}, $\sensing\in \R^{\numsam \times \dimension}$ is the sensing matrix, $\obs \in \R^{\numsam}$ are the responses (or observations), $\signal \in \R^\dimension$ is the loading vector and $\lambda \in \R^{++}$ governs the sparsity of the solution. Along with many efficient algorithms for its solution, the LASSO formulation is now backed with a rather mature theory for the generalization of its solutions as well as its variable selection consistency \cite{duchi2008efficient,bickel2009simultaneous,wainwright2009sharp,zhao2006model}.

While the long name attributed to \eqref{eq: LASSO} is apropos,\footnote{Many of the optimization solutions to LASSO leverage {\em shrinkage} operations  (e.g., as projections onto the $\ell_1$-ball) for {sparse} model {\em selections}. However, the geometric interplay of the $\ell_2$-data error objective and the $\ell_1$-norm constraint inherently promotes sparsity, independent of the algorithm.} it does not capture the LASSO's arbitrariness in subset selection via {shrinkage} to best explain the responses. In fact, this {\em uninformed} {selection} process not only prevents interpretability of results in many problems, but also fails to exploit key prior information that could radically improve learning performance. Based on this premise, approaches to guide the selection process of the LASSO are now aplenty.

Surprisingly, while the prior information in many regression problems generate fundamentally discrete constraints (e.g., on the sparsity patterns or the {\em support} of the LASSO solution), the majority of the existing approaches that enforce such constraints in selection are inherently continuous. For instance, a prevalent approach is to tailor a {\em sparsity inducing} norm to the constraints on the support set (c.f., \cite{jenatton2009structured}). That is, we create a structured convex norm by {\em mixing} basic norms with weights over pre-defined groups or using the Lov\'asz extension of non-decreasing submodular set functions of the support. As many basic norms have well-understood behavior in sparse selection, reverse engineering such norms is quite intuitive.

While such structure inducing, convex norm-based approaches on the LASSO are impressive, our contention in this paper is that, in order to truly make an impact in structured sparsity problems, we must fully leverage explicitly combinatorial approaches to guide LASSO's subset selection process. To achieve this, we show how Euclidean projections with structured sparsity constraints correspond to an integer linear program (ILP), which can be exactly or approximately solved subject to  matroid (via the greedy algorithm), and certain linear inequality constraints (via convex relaxation or multi-knapsack solvers).
A key actor in this process is a polynomial-time combinatorial algorithm that goes beyond simple selection heuristics towards provable solution quality as well as runtime/space bounds.

Furthermore, we introduce our combinatorial selection and least absolute shrinkage (\clash) operator and theoretically characterize its estimation guarantees. \clash enhances the {\em model-based compressive sensing} (model-CS) framework \cite{modelCS} by combining $\ell_1$-norm and combinatorial constraints on the regression vector. Therefore, \clash uses a combination of shrinkage and hard thresholding operations to significantly outperform the model-CS approach, LASSO, or continuous structured sparsity approaches in learning performance of sparse linear models. Furthermore, \clash establishes a regression framework where the underlying tractability of approximation in combinatorial selection is directly reflected in the algorithm's estimation and convergence guarantees.

The organization of the paper is as follows. In Sections \ref{sec: prelim} and \ref{sec: prox}, we set up the notation and the exact projections with structured sparsity constraints. We develop \clash in Section \ref{sec: algo} and highlight the key components of its convergence proof in Section \ref{sec:ing}. We present numerical results in Section \ref{sec: experiments}. We provide our conclusions in Section \ref{sec: conc}.

\section{Preliminaries}\label{sec: prelim}
\textbf{Notation:}
We use $[\signal]_j $ to denote the $j$-th element of $\signal$, and let $ \signal_i $ represent the $i$-th iterate of \clash.  The index set of $ \dimension $ dimensions is denoted as $ \mathcal{N} = \lbrace 1, 2, \dots, \dimension \rbrace $. Given $ \mathcal{S} \subseteq \mathcal{N} $, we define the complement set $ \mathcal{S}^c = \mathcal{N}\setminus \mathcal{S} $. Moreover, given a set $ \mathcal{S} \subseteq  \mathcal{N}$ and a vector $ \signal \in \mathbb{R}^\dimension $, $ (\signal)_{\mathcal{S}} \in \mathbb{R}^\dimension $ denotes a vector with the following properties: $ \left[(\signal)_{\mathcal{S}}\right]_{\mathcal{S}}=\left[\signal\right]_{\mathcal{S}}$ and $\left[(\signal)_{\mathcal{S}}\right]_{\mathcal{S}^c}= 0 $. The support set of $ \signal $ is defined as ${\rm supp}(\signal) = \lbrace i: [\signal]_i \neq 0 \rbrace $. We use $\abs{\mathcal{S}}$ to denote the cardinality of the set $\mathcal{S}$. The empirical data error is denoted as $f(\signal) \triangleq \vectornorm{\obs - \sensing \signal}_2^2 $, with gradient defined as $ \nabla f(\signal)  \triangleq -2 \sensing^T(\obs - \sensing \signal)$,  where $^T $ is the transpose operation. The notation $ \nabla_{\mathcal{S}}f(\signal)$ is shorthand for $\left(\nabla f(\signal)\right)_{\mathcal{S}}$. $ \mathbb{I} $ represents the identity matrix. % with dimensions apparent from the context. %Finally, please refer to the supplementary material for some definitions of basic combinatorial concepts.

\textbf{Combinatorial notions of sparsity:} We provide some definitions on combinatorial sparse models, and elaborate on a subset of interesting models with algorithmic implications.
\begin{definition}[Combinatorial sparsity models (CSMs)] \label{def: CSMs}
  We define a combinatorial sparsity model $\constraint  = \lbrace \mathcal{S}_m:$ $ \forall m,~{\mathcal{S}}_m \subseteq  \mathcal{N},~ \abs{{\mathcal{S}}_m} \le \sparsity \rbrace $ with the sparsity parameter $\sparsity$ as a collection of distinct index subsets ${\mathcal{S}}_m$.
\end{definition}

Throughout the paper, we assume that any CSM $\constraint$ is downward compatible, i.e., removing any subset of indices of any given element in $\constraint$, it is still in $\constraint$.

\textbf{Properties of the regression matrix:} Deriving approximation guarantees for \clash behooves us to assume the restricted isometry property (RIP) (defined below) on the regression matrix $\sensing$ \cite{Candes04A}. While the RIP and other similar conditions for deriving consistency properties of LASSO and its variants, such as the unique/exact representation property or the irrepresentable condition \cite{tropp2006algorithms,zhao2006model,jenatton2009structured,jacob2009group,bach2010structured}, are unverifiable {\em a priori} without exhaustive search, many {\em random} matrices satisfy them with high probability. 

\begin{definition}[RIP \cite{Candes04A,modelCS}] The regression matrix has the $\sparsity$-RIP with an isometry constant $\delta_{\sparsity}$ when %\vspace{-0.2cm}
\begin{equation}
(1-\delta_{\sparsity}){\vectornorm{\signal}_2^2} \leq {\vectornorm{\sensing \signal}_2^2} \leq (1+\delta_{\sparsity}){\vectornorm{\signal}_2^2} \label{eq:RIP1}, %\vspace{-0.2cm}
\end{equation} $ \forall \emph{\text{supp}}(\signal) \in \mathcal{C}_{\sparsity} $, where $\delta_{\sparsity} = \max_{\mathcal{S} \in \mathcal{C}_{\sparsity}} \vectornormmed{\sensing_{\mathcal{S}}^T \sensing_{\mathcal{S}} - \mathbb{I}}_{2\rightarrow 2}$, and $\sensing_{\mathcal{S}}$ is a submatrix of $\sensing$ as column-indexed by ${\mathcal{S}}$.
\end{definition} 
Here, we also comment on the scaling of $(\sparsity,\numsam,\dimension)$ for the desired level of isometry. When the entries of $\sensing$ can be modeled as independent and identically distributed (iid) with respect to a sub-Gaussian distribution, we can show that $\numsam =\bigO{\left(\delta_{\sparsity}^{-2}(\log(2M) + \sparsity \log(12\delta_{\sparsity}^{-1}))\right)}$ with overwhelming probability \cite{modelCS}. Here, $M$ is the minimum number of subspaces covering ${\mathcal{C}_{\sparsity}}$. While $\numsam$ explicitly depends on $\dimension$, for certain restricted CSMs, such as the rooted connected tree of \cite{modelCS}, this dependence can be quite weak, e.g.,  $\numsam = \bigO{(\sparsity)}$.

\section{Exact and approximate projections onto CSMs}\label{sec: prox}
The workhorse of the model-CS approach is the following non-convex projection problem onto CSMs, as defined by $\constraint$, which is a basic subset selection problem: %\vspace{-0.2cm}
\begin{equation}
\mathcal{P}_{\constraint}(\signal) = \argmin_{w \in \mathbb{R}^{\dimension}}\left\{\vectornorm{w - \signal}_2^2: {\rm  supp}(w) \in \constraint \right\}, \label{eq:proj} %\vspace{-0.2cm}
\end{equation}
where $\mathcal{P}_{\constraint}(\signal)$ is the projection operator. \cite{modelCS} shows that as long as $\mathcal{P}_{\constraint}(\cdot) $ is exactly computed in polynomial time for a CSM, their sparse recovery algorithms inherit strong approximation guarantees for that CSM. To better identify the CSMs that live within the model-CS assumptions, we first state the following key observation---the proof can be found in \cite{KyrillidisCevherSketching}.
\begin{lemma}[Euclidean projections onto CSMs]{\label{lemma:mod}}
  The support of the Euclidean projection onto $\constraint$ in \eqref{eq:proj} can be obtained as a solution to the following discrete optimization problem: %\vspace{-0.2cm}
\begin{equation}\label{eq: approximate proj}
{\rm supp}\left(\mathcal{P}_{\constraint}(\signal) \right) = \argmax_{\mathcal{S}: \mathcal{S} \in \constraint}F(\mathcal{S};\signal), %\vspace{-0.2cm}
\end{equation}
where $ F(\mathcal{S};\signal)= \|\signal\|_2^2 - \|(\signal)_\mathcal{S}-\signal\|_2^2 = \sum_{i\in \mathcal{S}} \abs{[\signal]_i}^2$ is the modular, variance reduction set function. Moreover, let $\widehat{\mathcal{S}}\in \constraint$ be the minimizer of the discrete problem. Then, it holds that $\mathcal{P}_{\constraint}(\signal)= (\signal)_{\widehat{\mathcal{S}}}$, which corresponds to hard thresholding.
\end{lemma}

The following proposition refines this observation to further accentuate the algorithmic implications for CSMs:
\begin{prop}[CSM projections via ILP's]{\label{prop:int}} The problem \eqref{eq: approximate proj} is equivalent to the following integer linear program (ILP): %\vspace{-0.2cm}
\begin{equation}\label{eq: ILP}
{\rm supp}\argmin_{\substack{z: [z]_i \in \lbrace 0,1 \rbrace, \\ \text{supp}(z) \in \constraint}}\big\lbrace w^Tz: [w]_i = -|[\signal]_i|^2 \big \rbrace, %\vspace{-0.2cm}
\end{equation} where $[z]_i$,  ($i = 1, \dots, n$), are support indicator variables.
\end{prop}
\noindent The proof of Proposition \ref{prop:int} is straightforward and is omitted.

Regardless of whether we use a dynamic program, a greedy combinatorial algorithm, or an ILP solver, the formulations \eqref{eq: approximate proj} or \eqref{eq: ILP} make the underlying tractability of the combinatorial selection explicit. We highlight this notion via the polynomial-time modular $\epsilon$-approximation property ($\text{PMAP}_{\epsilon}$):

\begin{definition}[PMAP$_\epsilon$ \cite{KyrillidisCevherSketching}] A CSM has the PMAP$_\epsilon$ with constant $\epsilon$, if the modular subset selection problem \eqref{eq: approximate proj} or the ILP \eqref{eq: ILP}  admit an $\epsilon$-approximation scheme with polynomial or pseudo-polynomial time complexity as a function of $\dimension$, $\forall \signal \in \R^\dimension$. Denoting the $\epsilon$-approximate solution of \eqref{eq: approximate proj} or \eqref{eq: ILP} as $\widehat{\mathcal{S}}_\epsilon$, this means $F(\widehat{\mathcal{S}}_\epsilon;\signal) \ge (1-\epsilon) \max_{\mathcal{S} \in \constraint} F(\mathcal{S};\signal)$.
\end{definition}

In this paper, we focus and elaborate on CSMs with PMAP$_0$. \vspace{-0.2cm} % where model-CS approach is applicable.\qq
\subsection{Example CSMs with PMAP$_0$}
\textbf{Matroids:} By matroid, we mean that $\constraint = (\mathcal{N}, \mathcal{I})$ is a finite collection of subsets of $\mathcal{N}$ that satisfies three conditions: $(i)$ $\lbrace \emptyset \rbrace \in \mathcal{I}$, $(ii)$ if $\mathcal{S}$ is in $\mathcal{I}$, then any subset of $\mathcal{S}$ is also in $\mathcal{I}$, and (iii) for $\mathcal{S}_1,\mathcal{S}_2 \in \mathcal{I}$ and $\abs{\mathcal{S}_1}>\abs{\mathcal{S}_2}$, there is an element ${s}\in \mathcal{S}_1\setminus\mathcal{S}_2$ such that $\mathcal{S}_2\cup \{s\}$ is in $\mathcal{I}$. As a simple example, the unstructured sparsity model (i.e., $x$ is $\sparsity$-sparse) forms a \emph{uniform matroid} as it is defined as the union of all subsets of $\mathcal{N}$  with cardinality $\sparsity$ or less. When $\constraint$ forms a matroid, the greedy basis algorithm can efficiently compute \eqref{eq:proj} by solving \eqref{eq: approximate proj} \cite{nemhauser1988integer} where sorting and selecting the $\sparsity$ largest elements in absolute value is suffcient to obtain the exact projection.  

Moreover, it turns out that this particular perspective provides a principled and tractable approach to encode an interesting class of \emph{matroid-structured sparsity models}. The recipe is quite simple: we seek the intersection of a structure provider matroid (e.g., partition, cographic/graphic, disjoint path, or matching matroid) with the sparsity provider uniform matroid. While the intersection of two matroids is not a matroid in general, we can prove that the intersection of the uniform matroid with any other matroid satisfies the conditions above.

\textbf{Linear support constraints:} Many interesting CSMs  $\constraint$ can be encoded using  {\it linear support constraints} of the form:
\begin{align}\nonumber
\constraint = \bigcup_{\forall z\in \mathfrak{Z}} {\rm supp}\left(z\right) ,~\mathfrak{Z}:=\left\{[z]_i\in \{0,1\}:  A z \leq b \right\},
\end{align} where $[A,b]$ is an integral matrix, and the first row of $A$ is all 1's and $[b]_1=\sparsity$. As a basic example, the neuronal spike model of \cite{hegde2009compressive} is based on linear support constraints where each spike respects a minimum refractory distance to each other. 

A key observation is that if each of the nonempty faces of $\mathfrak{Z}$ contains an integral point (i.e., forming an integral polyhedra), then convex relaxation methods can {\it exactly} obtain the correct integer solutions in polynomial time. In general, checking the integrality of $\mathfrak{Z}$ is NP-Hard. However, if $\mathfrak{Z}$ is integral and non-empty for all integral $b$, then a necessary condition is that $A$ be a totally unimodular (TU) matrix \cite{nemhauser1988integer}. A matrix is totally unimodular if the determinant of each square submatrix is equal to 0,1, or -1. Example TU matrices include interval, perfect, and network matrices \cite{nemhauser1988integer}. As expected, the constraint matrix $A$ of \cite{hegde2009compressive} is TU. Moreover, it is easy to verify that the sparse disjoint group model of \cite{friedman2010note} also defines a TU constraint, where groups have individual sparsity budgets.

\subsection{How about PMAP$_{\epsilon}$?}
For completeness and due to lack of space, we only mention PMAP$_{\epsilon}$, which extends the breath of the model-CS approach. For a detailed treatment of PMAP$_\epsilon$ and \clash, c.f. \cite{KyrillidisCevherSketching}, which describes multi-knapsack CSMs as a concrete example.  Moreover, for many of the PMAP$_0$ examples above, we can employ $\epsilon$-approximate---randomized---techniques to reduce computational cost. 

\section{The \clash algorithm}\label{sec: algo}

The new \clash algorithm obtains approximate solutions to the LASSO problem in \eqref{eq: LASSO} with the added twist that the solution must live within the CSM, as defined by $\constraint$: %\vspace{-0.2cm}
\begin{align}\label{eq: cLASSO}
\widehat{\signal}_{\text{\clash}} = \arg\min \big \lbrace f(x): \|\signal\|_1\le \lambda, \text{supp}(\signal) \in \constraint \big \rbrace. %\vspace{-0.2cm}
\end{align}
When available, using the CSM constraint $\constraint$ in addition to the $\ell_1$-norm constraint enhances learning in two important ways. First, the combinatorial constraints restricts the LASSO solution to exhibit true model-based supports, {\it increasing the interpretability} of the solution without relaxing $\constraint$ into a convex norm. Second, it empirically requires much fewer number of samples to obtain the true solution than both the LASSO and the model-CS approaches.

We provide a pseudo-code of an example implementation of \clash in Algorithm \ref{algo: class}. One can think of alternative ways of implementing \clash, such as single gradient updates in Step 2, or removing Step 4 altogether. While such changes may lead to different---possibly better---approximation guarantees for the solution of \eqref{eq: cLASSO}, we observe degradation in the empirical performance of the algorithm as compared to this implementation, whose guarantees are as follows:

\begin{theorem}[Iteration invariant]\label{thm: iteration invariant} Let $ \bestsignal \in \mathbb{R}^\dimension $ be the true vector that satisfies the constraints of \eqref{eq: cLASSO} and let $\obs = \sensing \bestsignal + \noise $ be the set of observations with additive error $\noise \in \mathbb{R}^{\numsam}$. Then,  the $i$-th iterate $\signal_i$ of \clash satisfies the following recursion: %\vspace{-0.1cm}
\begin{align}
\vectornorm{\signal_{i+1} - \bestsignal}_2 &\leq \rho \vectornorm{\signal_i - \bestsignal}_2 + c_1(\delta_{2\sparsity}, \delta_{3\sparsity})\vectornorm{\noise}_2  \nonumber \vspace{-0.2cm}
\end{align}
where $ \rho \triangleq \frac{\delta_{3\sparsity} + \delta_{2\sparsity}}{\sqrt{1-\delta_{2\sparsity}^2}} \sqrt{\frac{1+3\delta_{3\sparsity}^2}{1-\delta_{3\sparsity}^2}}  $ and $ c_1(\delta_{2\sparsity}, \delta_{3\sparsity}) $ is a constant defined in \cite{clash}. The iterations contract when $\delta_{3\sparsity} < 0.3658$.
\end{theorem}

Theorem \ref{thm: iteration invariant} shows that the isometry requirements of \clash are competitive with the mainstream hard thresholding methods, such as CoSaMP \cite{cosamp} and Subspace Pursuit \cite{SP}, even though it incorporates the $\ell_1$-norm constraints, which, as Section \ref{sec: experiments} illustrates, improves learning performance.

\begin{algorithm}[t]
   \caption{\clash Algorithm}\label{algo: class}
\begin{algorithmic}[1]
   \Statex {\bfseries Input:} $\obs$, $\sensing$, $ \lambda $, $\mathcal{P}_{\constraint}$, Tolerance $ \eta $, MaxIterations
   \Statex {\bfseries Initialize:} $ \signal_0 \leftarrow 0 $, $ \mathcal{X}_0 \leftarrow \lbrace \emptyset \rbrace $, $ i \leftarrow 0 $
   \Statex {\bfseries repeat} %\REPEAT
   \State \hspace{0.16cm} $ \mathcal{S}_i \leftarrow \text{supp}(\mathcal{P}_{\constraint}(\nabla_{\mathcal{X}_i^c} f(\signal_i))) \cup \mathcal{X}_i$
   \State \hspace{0.16cm} $ v_i \leftarrow \argmin_{v: \vectornorm{v}_1 \leq \lambda, \;\text{supp}(v) \in \mathcal{S}_i} \vectornorm{\obs - \sensing v}_2^2 $
   \State \hspace{0.16cm} $ \gamma_i  \leftarrow \mathcal{P}_{\constraint}(v_i)$ with $ \Gamma_i \leftarrow \text{supp}(\gamma_i) $
   \State \hspace{0.16cm} $ \signal_{i+1} \leftarrow \argmin_{\signal: \vectornorm{\signal}_1 \leq \lambda, \;\text{supp}(\signal) \in \Gamma_i} \vectornorm{\obs - \sensing \signal}_2^2 $
   \State \hspace{0.16cm} $ \mathcal{X}_{i+1} \leftarrow \text{supp}(\signal_{i+1}) $
   \Statex $ i \leftarrow i + 1 $.
   \Statex {\bfseries until} $\vectornorm{\signal_i - \signal_{i-1}}_2 \leq \eta \vectornorm{\signal_i}_2 $ or MaxIterations.
\end{algorithmic}
\end{algorithm}

\begin{remark}\label{modelselect}[Model mismatch and selection] Let us assume a generative model $\obs = \sensing \beta + \tilde{\noise}$. Let $\bestsignal$ be the best approximation of $\beta$ in $\constraint$ within $\ell_1$-ball of radius $\lambda$. Then, we can show that the iteration invariant of Theorem \ref{thm: iteration invariant} still holds with $\text{SNR}= \frac{\|\bestsignal\|_2}{\|{\noise}\|_2}$, where
$\|{\noise}\|_2 \le \|\tilde{\noise}\|_2 + \|\sensing(\beta-\bestsignal)\|_2$, where the latter quantity (the impact of mismatch) can be analyzed using the restricted amplification property of $\sensing$ \cite{modelCS}. For instance, when $\constraint$ is the uniform sparsity model, then  $\|\sensing(\beta-\bestsignal)\|_2 \le \sqrt{1+\delta_\sparsity}\left(\|\beta-\bestsignal\|_2 +  \frac{\|\beta-\bestsignal\|_1}{\sqrt{\sparsity}} \right)$, which should presumably be small if the model is selected correctly. 
\end{remark}

In the absence of prior information, we automate the parameter selection by using the Donoho-Tanner phase transition \cite{donoho2005nrp} to choose the maximum $\sparsity$ allowed for a given $(\numsam,\dimension)$-pair, and then by using cross validation to pick $\lambda$ \cite{ward2009compressed}.

\section{Proof of Theorem \ref{thm: iteration invariant}}{\label{sec:ing}}
We sketch the proof of Theorem \ref{thm: iteration invariant}  a l\'a \cite{cosamp} and \cite{foucart2010sparse} assuming the general case of $\text{PMAP}_{\epsilon}$. The details of the proof can be found in an extended version of the paper \cite{clash}. 

\begin{lemma}[Active set expansion - Step 1]{\label{lemma:active_set_exp}}\textit{The support set $ \mathcal{S}_i $, where $ |\mathcal{S}_i| \leq 2\sparsity $, identifies a subspace in $ \constrainttwo $ such that:}
\begin{align}\nonumber
\vectornorm{(\signal_{i} - \bestsignal)_{\mathcal{S}_i^c}}_2 &\leq (\delta_{3\sparsity} + \delta_{2\sparsity} + \sqrt{\epsilon}(1+\delta_{2\sparsity}))\vectornorm{\signal_i - \bestsignal}_2 \nonumber \\ &+ \big(\sqrt{2(1+\delta_{3\sparsity})} + \sqrt{\epsilon(1+\delta_{2\sparsity}})\big)\vectornorm{\noise}_2
\end{align}
\end{lemma} Lemma \ref{lemma:active_set_exp} states that, at each iteration, Step 1 of \clash identifies a $2\sparsity$ support set such that the unrecovered energy of $\bestsignal$ is bounded. For $\epsilon = 0 $, \clash exactly identifies the support where the projected gradient onto $\constraint$ can make most impact on the loading vector in the support complement of its current solution, which are subsequently merged together.

\begin{lemma}[Greedy descent with least absolute shrinkage - Step 2]{\label{lemma:greedy}} \textit{Let $ \mathcal{S}_i $ be a $ 2\sparsity $-sparse support set. Then, the least squares solution $ v_i $ in step 2 of Algorithm 1 satisfies}
\begin{align}\nonumber
\vectornorm{v_i - \bestsignal}_2 \leq \frac{1}{\sqrt{1-\delta_{3\sparsity}^2}} \vectornorm{(\signal_i - \bestsignal)_{\mathcal{S}_i^c}}_2 + \frac{\sqrt{1+\delta_{2\sparsity}}}{1-\delta_{3\sparsity}} \vectornorm{\noise}_2. %\label{eq:15}
\end{align}
\end{lemma} We borrow the proof of Lemma \ref{lemma:greedy} from \cite{foucart2010sparse}. This step improves the objective function $f(\signal)$ as much as possible on the active set in order to arbitrate the active set. The solution simultaneously satisfies the $ \ell_1 $-norm constraint.

Step 3 projects the solution onto $\constraint$, whose action is characterized by the following lemma. Here, we show the $\epsilon$-approximate projection explicitly:
\begin{lemma}[Combinatorial selection - Step 3] {\label{lemma:comb_selection}} Let $ v_i $ be a $ 2\sparsity $-sparse proxy vector with indices in support set $ \mathcal{S}_i $, $ \constraint $ be a CSM and $ \gamma_i $ the projection of $ v_i $ under $ \constraint $. Then:
\begin{align}\nonumber
\vectornorm{\gamma_i  - v_i}_2^2 &\leq (1 - \epsilon) \vectornorm{(v_i - \bestsignal)_{\mathcal{S}_i}}_2^2 + \epsilon \vectornorm{v_i}_2^2. 
\end{align}
\end{lemma}

Step 4 requires the following Corollary to Lemma \ref{lemma:greedy}:
\begin{corollary}[De-bias - Step 4]\label{corollary:debias} Let $ \Gamma_i $ be the support set of a proxy vector $ \gamma_i $ where $ |\Gamma_i| \leq \sparsity $. Then, the least squares solution $ \signal_{i+1} $ in Step 4 satisfies
\begin{align}\nonumber
\vectornorm{\signal_{i+1} - \bestsignal}_2 \leq \frac{1}{\sqrt{1 - \delta_{2\sparsity}^2}} \vectornorm{\gamma_i  - \bestsignal}_2 + \frac{\sqrt{1+\delta_\sparsity}}{1 - \delta_{2\sparsity}} \vectornorm{\noise}_2. 
\end{align}
\end{corollary}

Step 4 de-biases the current result on the putative solution support. Its characterization connects Lemmas \ref{lemma:greedy} and \ref{lemma:comb_selection}:
\begin{lemma}\label{lemma:noname} Let $ v_i $ be the least squares solution of the greedy descent step (step 5) and  $ \gamma_i $ be a proxy vector to $ v_i $ after applying Combinatorial selection step. Then, $ \vectornorm{\gamma_i  - \bestsignal}_2 $ can be expressed in terms of the distance from $ v_i $ to $ \bestsignal $ as follows:
\begin{align}
&\vectornorm{\gamma_i  - \bestsignal}_2 \nonumber \\ 
&\leq \sqrt{1 + \big((1-\epsilon) + 2\sqrt{1-\epsilon}\big)\delta_{3\sparsity}^2 + 2\delta_{3\sparsity}\sqrt{\epsilon} + \epsilon}\cdot \vectornorm{v_i - \bestsignal}_2 \nonumber \\ &+ D_1\vectornorm{\noise}_2 + D_2 \vectornorm{\bestsignal}_2 + D_3 \sqrt{\vectornorm{\bestsignal}_2 \vectornorm{\noise}_2}, \label{eq:74}
\end{align} \textit{where $ D_1, D_2, D_3 $ are constants depending on $ \epsilon, \delta_{2\sparsity}, \delta_{3\sparsity} $.}
\end{lemma}

Finally, the proof of Theorem \ref{thm: iteration invariant} follows by concatenating Corollary \ref{corollary:debias} with Lemmas \ref{lemma:active_set_exp}, \ref{lemma:greedy}, and \ref{lemma:noname}, and setting $\epsilon = 0$. 

\section{Experiments}\label{sec: experiments}
In the following experiments, we compare algorithms from the following list: $ (i) $ the LASSO algorithm \cite{tibshirani96regression}, $ (ii) $ the Basis Pursuit DeNoising (BPDN) \cite{Chen98atomicdecomposition}, $ (iii) $ the sparse-\clash algorithm, where $\constraint $ is the index set of $\sparsity$-sparse signals, $ (iv) $ the model-\clash algorithm\footnote{\clash codes are available for MATLAB at {\it http://lions.epfl.ch/CLASH.}}, which explicitly carries $\constraint$, and $ (v) $  Subspace Pursuit (SP) algorithm \cite{SP}, as integrated with the model-CS approach.  We emphasize here that when $\lambda\rightarrow\infty$ in \eqref{eq: cLASSO}, \clash must converge to the model-based SP solution.

The LASSO algorithm finds a solution to the problem defined in \eqref{eq: LASSO}, where we use a Nesterov accelerated projected gradient algorithm. The BPDN algorithm in turn solves the following optimization problem: \vspace{-0.2cm}
\begin{align}{\label{eq: BPDN}}
\widehat{\signal}_{\text{BPDN}} = \arg\min \left\{\left\|\signal \right\|_1:~~ \vectornorm{\sensing \signal - \obs}_2 \leq \sigma \right\},\vspace{-0.2cm}
\end{align} where $ \sigma $ represents prior knowledge on the energy of the additive noise term. To solve (\ref{eq: BPDN}), we use the spectral projected gradient method SPGL1 algorithm \cite{BergFriedlander2008}.
 
In the experiments below, the nonzero coefficients of $ \bestsignal $ are generated iid according to the standard normal distribution with $\vectornorm{\bestsignal}_2 = 1$.  The BPDN algorithm is given the true $\sigma$ values. While \clash is given the true value of $\sparsity$ for the experiments below, additional experiments (not shown) shows that our phase transition heuristics is quite good and the mismatch is graceful as indicated in Remark 1. All the algorithms use a high precision stopping tolerance $\eta=10^{-5}$.

\begin{figure*}[!t]
\centering
\begin{tabular}{cc}
\centerline{\subfigure[]{\includegraphics[width = 0.44\textwidth]{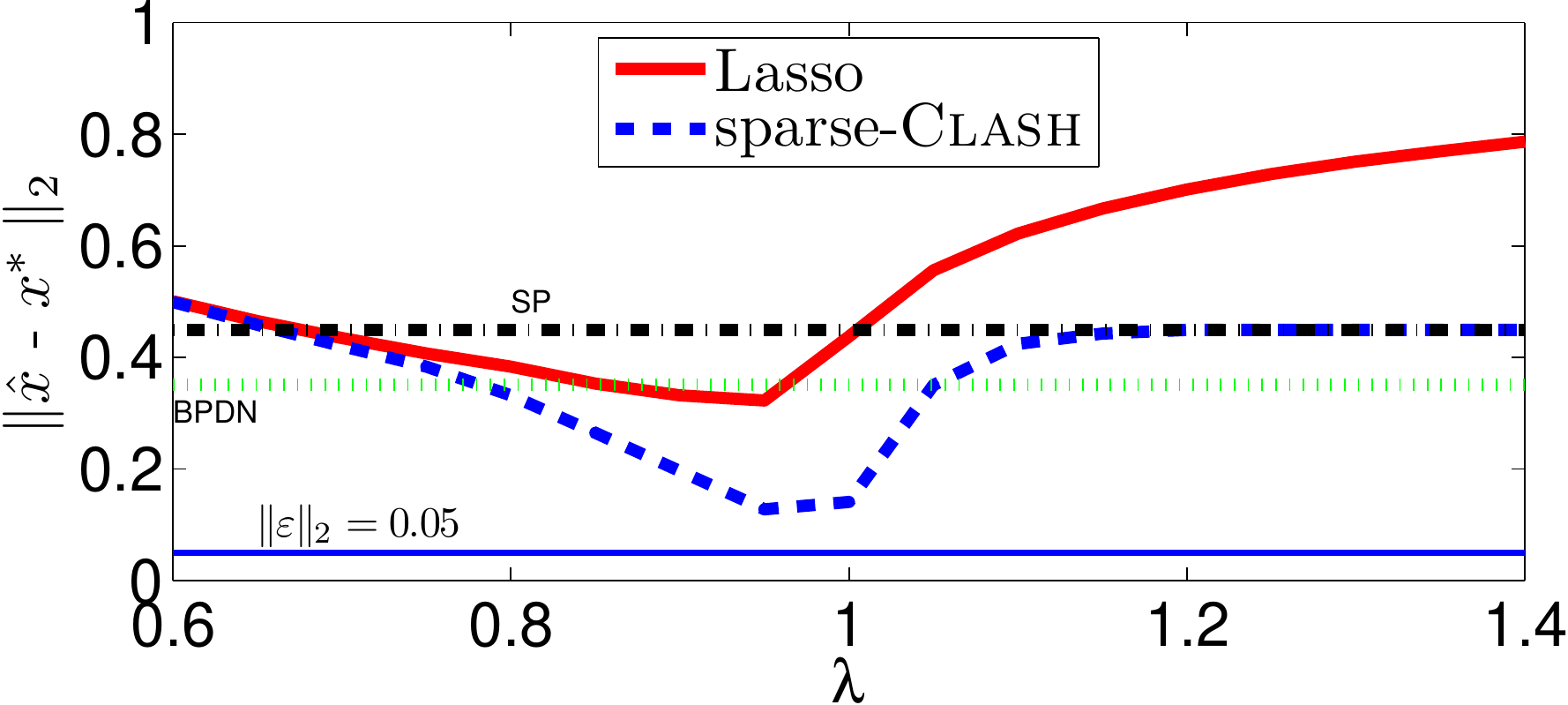}} \label{fig:1a}
\hfil
\subfigure[]{\includegraphics[width = 0.44\textwidth]{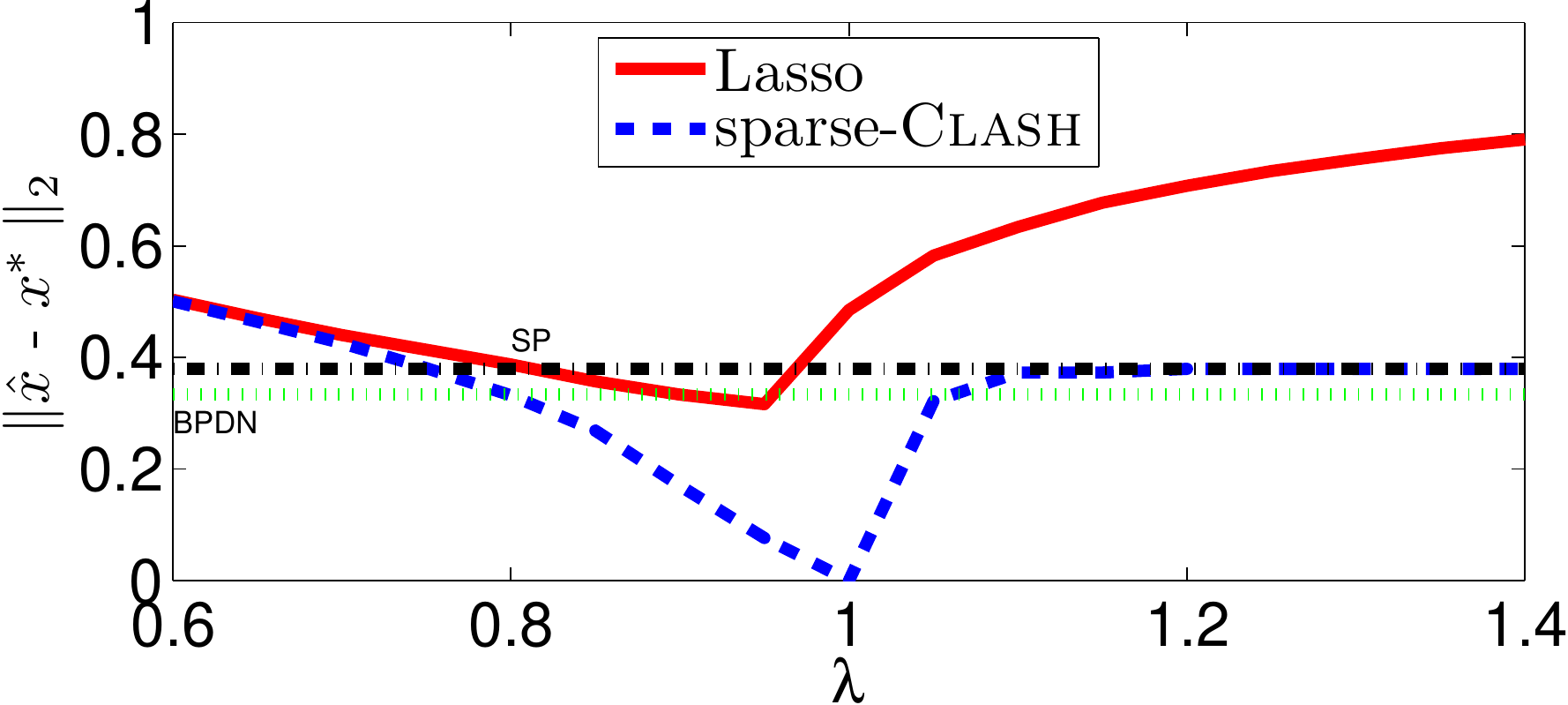}} \label{fig:1b} } \\
\centerline{\subfigure[]{\includegraphics[width = 0.44\textwidth]{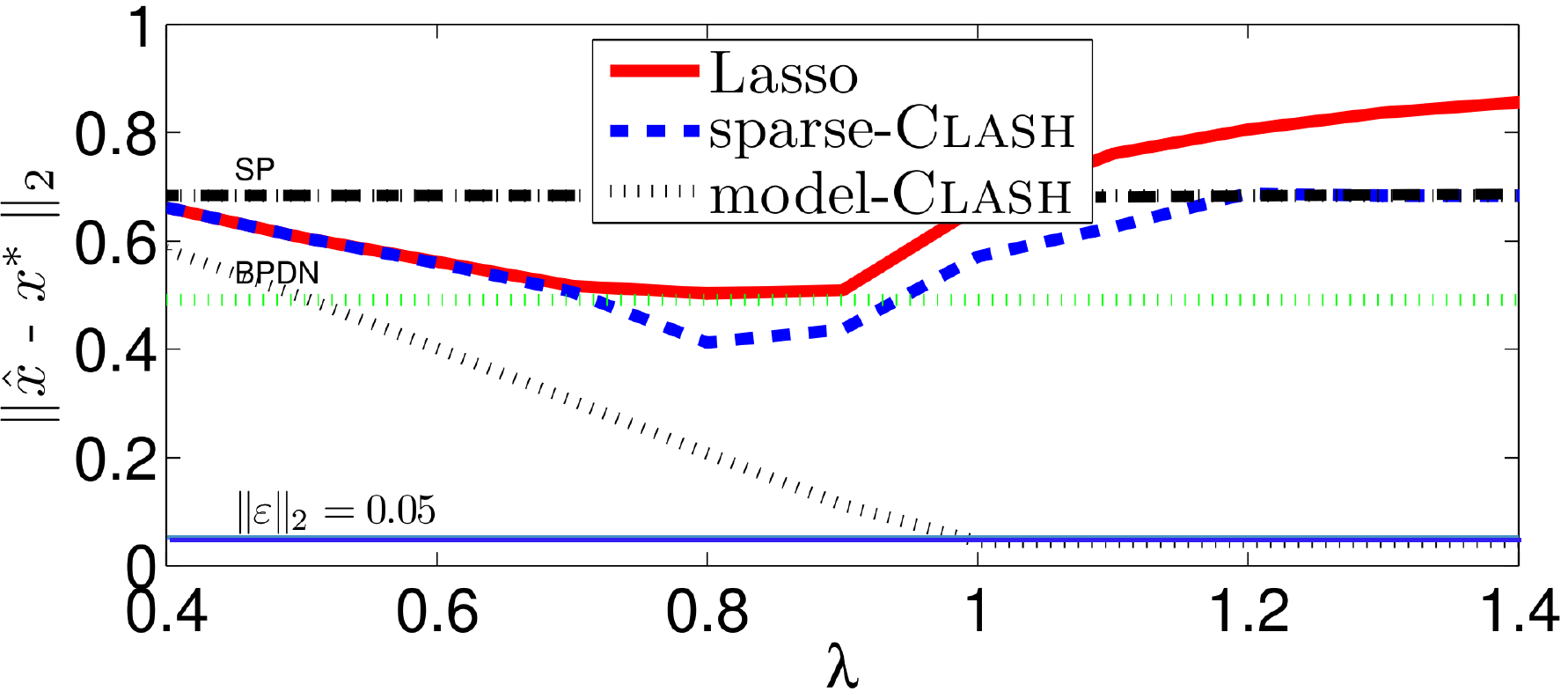}} \label{fig:1c}
\hfil
\subfigure[]{\includegraphics[width = 0.44\textwidth]{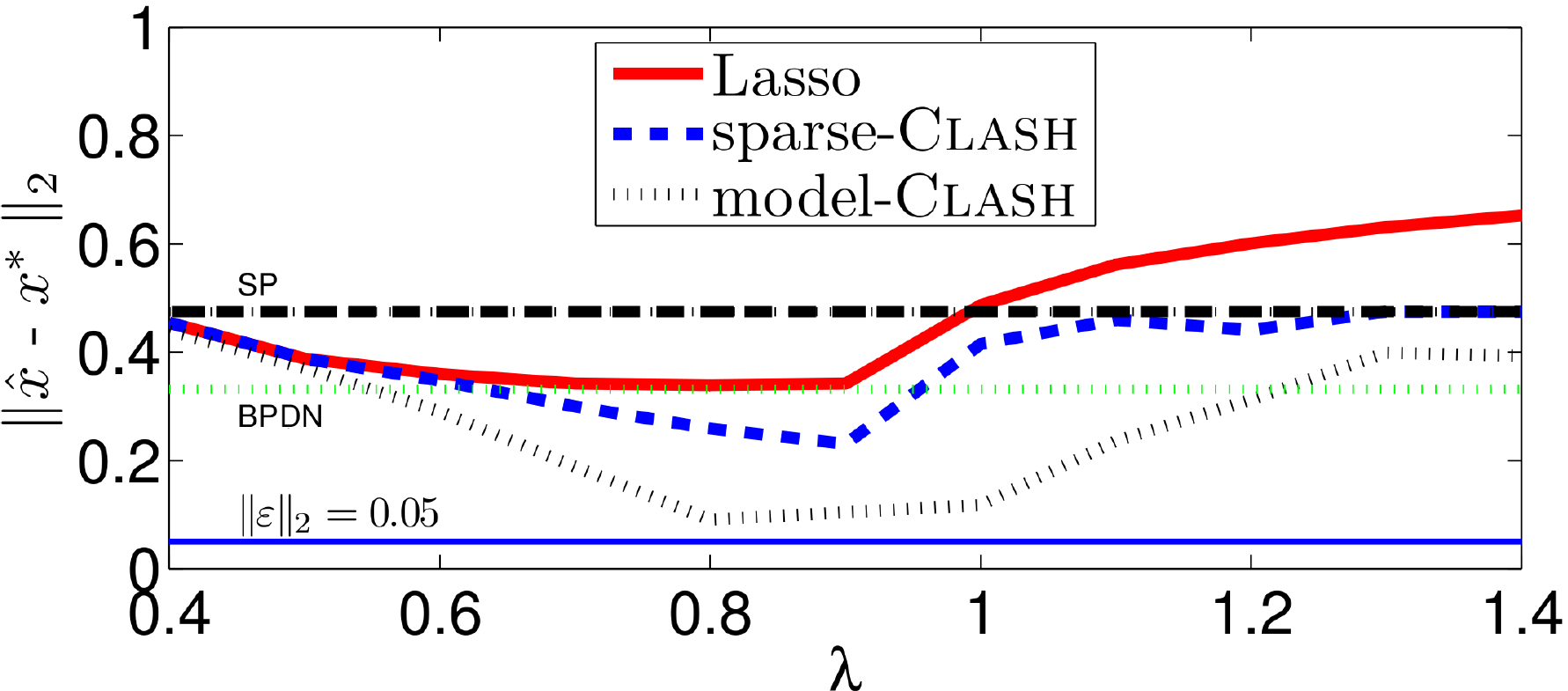}}\label{fig:1d}}
\end{tabular}
\caption{\small\sl
Median values of signal error $\vectornorm{\hat{\signal} - \bestsignal}_2$. Top row: simple sparsity model under noisy $ \vectornorm{\noise}_2 = 0.05 $ (left column) and noiseless $ \vectornorm{\noise}_2 = 0 $ (right column) settings. Bottom row: the $(\sparsity, C)$-clustered sparsity model (left column) and the TU model (right column).  }
\end{figure*}

\textbf{Experiment 1: Improving simple sparse recovery.}
In this experiment, we generate random realizations of the model $ \obs = \sensing \bestsignal + \noise $ for $ \dimension = 800 $. Here, $\sensing$ is a dense random matrix whose entries are iid Gaussian with zero mean and variance $1/\numsam$. We consider two distinct generative model settings: $ (i) $ with additive Gaussian white noise with $ \vectornorm{\noise}_2 = 0.05 $, $\numsam = 240 $ and $\sparsity = 89 $, and $(ii)$ the noiseless model ($\vectornorm{\noise}_2 = 0 $), $\numsam = 250 $ and sparsity parameter $\sparsity = 93$. For this experiment, we perform 500 Monte Carlo model realizations.

We sweep $\lambda$ and illustrate the recovery performance of \clash (\ref{eq: cLASSO}). Figures 1(a)-(b) illustrate that the combination of hard thresholding with norm constraints can {\it improve} the signal recovery performance significantly over convex-only and hard thresholding-only methods---both in noisy and noiseless problem settings. For $\vectornorm{\noise} = 0 $, \clash perfectly recovers the signal when $\lambda $ is close to the true value. When $\lambda\ll\vectornorm{\bestsignal}_1$, the performance degrades due to the large norm mismatch.

\textbf{Experiment 2: Improving structured sparse recovery} We consider two signal CSMs: in the first model, we assume $ \sparsity $-sparse signals that admit clustered sparsity with coefficients in $ C $-contiguous blocks on an undirected, acyclic chain graph \cite{cevher2009recovery}. Without loss of generality, we use $ C = 5 $ (Figure 1(c)). The second model corresponds to a TU system where we partition the $ \sparsity $-sparse signals into uniform blocks and force sparsity constraints on individual blocks; in this case, we solve the set optimization problem optimally via linear programming relaxation (Figure 1(d)).  Here, the noise energy level satisfies $ \vectornorm{\noise}_2 = 0.05$, and $ \dimension = 500 $, $ \numsam = 125 $, and $ \sparsity = 50$. In both cases, we conduct 100 Monte Carlo iterations and perform sparse estimation for a range of $\lambda$ values. 

In Figure 1(c), we observe that clustered sparsity structure provides a distinct advantage in reconstruction compared to LASSO formulation and the sparse-\clash algorithm. Furthermore, note that when $\lambda$ is large, norm constraints have {\it no effect} and the model-\clash provides essentially the same results as the model-CS approach \cite{modelCS}. On the other hand, the sparse-\clash improves significantly beyond the LASSO solution thanks to the $\ell_1$-norm constraint.

In Figure 1(d) however, the situation is radically changed: while the TU constraints enhance the reconstruction of model-CS approach over simple sparse recovery, the improvement becomes quite large as the $\ell_1$-norm constraint kicks in. We also observe the improvement in sparse-\clash but it is not as accentuated as the model-\clash.

\section{Conclusions}\label{sec: conc}
\clash establishes a regression framework where efficient algorithms from combinatorial and convex optimization can interface for interpretable and model-based sparse solutions. Our experiments demonstrate that while the model-based combinatorial selection by itself can greatly improve sparse recovery over the approaches based on uniform sparsity alone, the shrinkage operations due to the $\ell_1$-constraint has an undeniable, positive impact on the learning performance. Understanding the tradeoffs between the complexity of approximation and the recovery guarantees of \clash in this setting is a promising theoretical as well as practical direction.

\appendix
\subsection{Proof of Theorem 1}

A well-known lemma used in the convergence guarantee proof of \clash is defined next. The proof is omitted.

\begin{lemma}[Optimality condition] \textit{Let $ \Theta \subseteq \mathbb{R}^\dimension $ be a convex set and $ f: \Theta \rightarrow \mathbb{R} $ be a smooth objective function defined over $ \Theta $. Let $ \psi^\ast \in \Theta $ be a local minimum of the objective function $ f $ over the set $ \Theta $. Then}
\begin{align}
\langle \nabla f(\psi^\ast), \psi - \psi^\ast \rangle \geq 0, \;\; \forall \psi \in \Theta, 
\end{align} \textit{for all convex sets $ \Theta $.}
\end{lemma}

In the derivation of Theorem 1, we assume $ \bestsignal \in \mathbb{R}^{\dimension} $ is the loading vector, $ \obs \in \mathbb{R}^{\numsam} $ is the set of observations, $ \sensing \in \mathbb{R}^{\numsam \times \dimension} $ is the regression matrix and $ \noise = \obs - \sensing \bestsignal $ represents the additive noise term. For clarity reasons, we present the proof of Theorem 1 as a collection of lemmas to help readability.

\begin{lemma}[Active set expansion] \textit{The support set $ \mathcal{S}_i $, where $ |\mathcal{S}_i| \leq 2\sparsity $, identifies a subspace in $ \constrainttwo $ such that:}
\begin{align}
\vectornorm{(\signal_{i} - \bestsignal)_{\mathcal{S}_i^c}}_2 \leq (\delta_{3\sparsity} + \delta_{2\sparsity} + \sqrt{\epsilon}(1+\delta_{2\sparsity}))\vectornorm{\signal_i - \bestsignal}_2 + \big(\sqrt{2(1+\delta_{3\sparsity})} + \sqrt{\epsilon(1+\delta_{2\sparsity}})\big)\vectornorm{\noise}_2. \label{eq:99a}
\end{align} 
\end{lemma}

\begin{proof} 
Let $ \mathcal{X}_i \cup \mathcal{X}^\ast $ denote the union of the support sets of the current estimate $ \signal_i $ and the signal of interest $ \bestsignal $. Then, the following sequence of inequalities hold true:
\begin{align}
F(\mathcal{X}_i \cup \mathcal{X}^\ast ; \nabla f(\signal_i)) &\leq F(\mathcal{X}_i \cup \text{supp}(\mathcal{P}_{\constraint}(\nabla_{\mathcal{X}_i^c} f(\signal_i))); \nabla f(\signal_i)) \label{eq:901} \Rightarrow \\
(1-\epsilon)F(\mathcal{X}_i \cup \mathcal{X}^\ast; \nabla f(\signal_i)) &\leq (1-\epsilon) F(\mathcal{X}_i \cup \text{supp}(\mathcal{P}_{\constraint}(\nabla_{\mathcal{X}_i^c} f(\signal_i))); \nabla f(\signal_i)) \label{eq:902}
\end{align} Given that support set $ \mathcal{S}_i $ is an $ \epsilon $-approximate support set, from the definition of PMAP, (\ref{eq:902}) is further transformed into:
\begin{align}
(1-\epsilon)F(\mathcal{X}_i \cup \mathcal{X}^\ast; \nabla f(\signal_i)) &\leq F(\mathcal{S}_i ; \nabla f(\signal_i)). \label{eq:92}
\end{align} Substituting the definition of the variance reduction modular function $ F(\mathcal{S};\signal)\triangleq \|\signal\|_2^2 - \|(\signal)_\mathcal{S}-\signal\|_2^2 = \vectornorm{(\signal)_{\mathcal{S}}}_2^2$, we get:
\begin{align}
(1-\epsilon)\vectornormmed{\nabla_{\mathcal{X}_i \cup \mathcal{X}^\ast} f(\signal_i)}_2^2&\leq \vectornormmed{\nabla_{\mathcal{S}_i} f(\signal_i)}_2^2 \Rightarrow \\
(1-\epsilon)\vectornormbig{\Big(\sensing^\ast (\obs - \sensing \signal_i)\Big)_{\mathcal{X}_i \cup \mathcal{X}^\ast}}_2^2&\leq \vectornormbig{\Big(\sensing^\ast (\obs - \sensing \signal_i)\Big)_{\mathcal{S}_i}}_2^2 \Rightarrow \\
\vectornormbig{\Big(\sensing^\ast (\obs - \sensing \signal_i)\Big)_{\mathcal{X}_i \cup \mathcal{X}^\ast}}_2^2 &\leq \vectornormbig{\Big(\sensing^\ast (\obs - \sensing \signal_i)\Big)_{\mathcal{S}_i }}_2^2 + \epsilon \vectornormbig{\Big(\sensing^\ast (\obs - \sensing \signal_i)\Big)_{\mathcal{X}_i \cup \mathcal{X}^\ast}}_2^2. \label{eq:93} 
\end{align} Using the subadditivity property of the square root function and excluding the common distribution $ \big(\sensing^\ast (\obs - \sensing \signal_i)\big)_{(\mathcal{X}_i \cup \mathcal{X}^\ast) \cap \mathcal{S}_i} $, we have:

\begin{align}
\vectornormbig{\Big(\sensing^\ast (\obs - \sensing \signal_i)\Big)_{(\mathcal{X}_i \cup \mathcal{X}^\ast)\setminus \mathcal{S}_i }}_2 &\leq \vectornormbig{\Big(\sensing^\ast (\obs - \sensing \signal_i)\Big)_{\mathcal{S}_i \setminus (\mathcal{X}_i \cup \mathcal{X}^\ast)}}_2 + \sqrt{\epsilon} \vectornormbig{\Big(\sensing^\ast (\obs - \sensing \signal_i)\Big)_{\mathcal{X}_i \cup \mathcal{X}^\ast}}_2 \label{eq:94} \\
 &\stackrel{(i)}{\leq} \vectornormbig{\Big(\sensing^\ast \sensing (\bestsignal - \signal_i)\Big)_{\mathcal{S}_i \setminus (\mathcal{X}_i \cup \mathcal{X}^\ast)}}_2 + \vectornormbig{(\sensing^\ast \noise)_{\mathcal{S}_i \setminus (\mathcal{X}_i \cup \mathcal{X}^\ast)}}_2 \nonumber \\ &+ \sqrt{\epsilon} \vectornormbig{\Big(\sensing^\ast \sensing (\bestsignal - \signal_i)\Big)_{\mathcal{X}_i \cup \mathcal{X}^\ast}}_2 + \sqrt{\epsilon} \vectornormbig{(\sensing^\ast \noise)_{\mathcal{X}_i \cup \mathcal{X}^\ast}}_2 \label{eq:941} \\
 &\stackrel{(ii)}{=} \vectornormbig{\Big((\sensing^\ast \sensing - \mathbb{I})(\bestsignal - \signal_i)\Big)_{\mathcal{S}_i \setminus (\mathcal{X}_i \cup \mathcal{X}^\ast)}}_2 + \vectornormbig{(\sensing^\ast \noise)_{\mathcal{S}_i \setminus (\mathcal{X}_i \cup \mathcal{X}^\ast)}}_2 \nonumber \\ &+ \sqrt{\epsilon} \vectornormbig{\Big(\sensing^\ast \sensing (\bestsignal - \signal_i)\Big)_{\mathcal{X}_i \cup \mathcal{X}^\ast}}_2 + \sqrt{\epsilon} \vectornormbig{(\sensing^\ast \noise)_{\mathcal{X}_i \cup \mathcal{X}^\ast}}_2 \label{eq:942} \\
 &\stackrel{(iii)}{\leq} (\delta_{3\sparsity} + \sqrt{\epsilon}(1+\delta_{2\sparsity}))\vectornorm{\signal_i - \bestsignal}_2 + \vectornormbig{(\sensing^\ast \noise)_{\mathcal{S}_i \setminus (\mathcal{X}_i \cup \mathcal{X}^\ast)}}_2 \nonumber \\ &+ \sqrt{\epsilon} \vectornormbig{(\sensing^\ast \noise)_{\mathcal{X}_i \cup \mathcal{X}^\ast}}_2. \label{eq:95}
%\vectornorm{\nabla f(\signal_i)|_{\mathcal{X}_i \cup \mathcal{X}^\ast}}_2 \leq \vectornorm{\nabla f(\signal_i)|_{\mathcal{S}_i^{\text{opt}}}}_2
\end{align} where $ (i) $ is obtained by applying the triangle inequality, $ (ii) $ holds since $ (\bestsignal - \signal_i)_{\mathcal{S}_i \setminus (\mathcal{X}_i \cup \mathcal{X}^\ast)} = 0 $ and $ (iii) $ is due to Cauchy-Swartz inequality and isometry constant definition.

In addition, we can obtain a lower bound for $ \vectornorm{\big(\sensing^\ast (\obs - \sensing \signal_i)\big)_{(\mathcal{X}_i \cup \mathcal{X}^\ast)\setminus \mathcal{S}_i }}_2 $:
\begin{align}
\vectornormbig{\Big(\sensing^\ast (\obs - \sensing \signal_i)\Big)_{(\mathcal{X}_i \cup \mathcal{X}^\ast)\setminus \mathcal{S}_i }}_2 &= \vectornormbig{\Big(\sensing^\ast \sensing (\bestsignal - \signal_i)\Big)_{(\mathcal{X}_i \cup \mathcal{X}^\ast)\setminus \mathcal{S}_i } + (\sensing^\ast \noise)_{(\mathcal{X}_i \cup \mathcal{X}^\ast)\setminus \mathcal{S}_i }}_2 \label{eq:96} \\
&= \vectornormbig{\Big(\sensing^\ast \sensing (\bestsignal - \signal_i)\Big)_{(\mathcal{X}_i \cup \mathcal{X}^\ast)\setminus \mathcal{S}_i } + (\bestsignal - \signal_i)_{(\mathcal{X}_i \cup \mathcal{X}^\ast)\setminus \mathcal{S}_i } \nonumber \\ &- (\bestsignal - \signal_i)_{(\mathcal{X}_i \cup \mathcal{X}^\ast)\setminus \mathcal{S}_i } + (\sensing^\ast \noise)_{(\mathcal{X}_i \cup \mathcal{X}^\ast)\setminus \mathcal{S}_i }}_2 \label{eq:96} \\
&\geq \vectornorm{(\bestsignal - \signal_i)_{(\mathcal{X}_i \cup \mathcal{X}^\ast)\setminus \mathcal{S}_i }}_2 - \vectornormbig{\Big((\sensing^\ast \sensing - \mathbb{I})(\bestsignal - \signal_i)\Big)_{(\mathcal{X}_i \cup \mathcal{X}^\ast)\setminus \mathcal{S}_i }}_2 \nonumber \\ &- \vectornorm{(\sensing^\ast \noise)_{(\mathcal{X}_i \cup \mathcal{X}^\ast)\setminus \mathcal{S}_i }}_2 \label{eq:97} \\
&\stackrel{(i)}{\geq} \vectornorm{(\bestsignal - \signal_i)_{(\mathcal{X}_i \cup \mathcal{X}^\ast)\setminus \mathcal{S}_i }}_2 - \delta_{2\sparsity} \vectornorm{\bestsignal - \signal_i}_2 - \vectornorm{(\sensing^\ast \noise)_{(\mathcal{X}_i \cup \mathcal{X}^\ast)\setminus \mathcal{S}_i }}_2. \label{eq:98}
\end{align} where $ (i) $ is obtained by using Cauchy-Swartz inequality and isometry constant definition.

Since $ \vectornorm{(\signal_i - \bestsignal)_{(\mathcal{X}_i \cup \mathcal{X}^\ast)\setminus \mathcal{S}_i }}_2 = \vectornorm{(\signal_i - \bestsignal)_{\mathcal{S}_i^c }}_2 $, combining (\ref{eq:95}) and (\ref{eq:98}), we get:
\begin{align}
\vectornorm{(\signal_i - \bestsignal)_{\mathcal{S}_i^c }}_2 &\leq (\delta_{3\sparsity} + \delta_{2\sparsity} + \sqrt{\epsilon}(1+\delta_{2\sparsity}))\vectornorm{\signal_i - \bestsignal}_2 + \big(\sqrt{2(1+\delta_{3\sparsity})} + \sqrt{\epsilon(1+\delta_{2\sparsity}})\big)\vectornorm{\noise}_2 \label{eq:99b}.
\end{align} as a consequence of the RIP inequality.
\end{proof}

\begin{lemma}{\label{lemma:greedy}}[Greedy descent with least absolute shrinkage] \textit{Let $ \mathcal{S}_i $ be a $ 2\sparsity $-sparse support set. Then, the least squares solution $ v_i $ given by:}
\begin{align}
v_i \leftarrow \argmin_{v: \vectornorm{v}_1 \leq \lambda, \text{supp}(v) \in \mathcal{S}_i} \vectornorm{\obs - \sensing v}_2^2, 
\end{align} \textit{satisfies:}
\begin{align}
\vectornorm{v_i - \bestsignal}_2 \leq \frac{1}{\sqrt{1-\delta_{3\sparsity}^2}} \vectornorm{(v_i - \bestsignal)_{\mathcal{S}_i^c}}_2 + \frac{\sqrt{1+\delta_{2\sparsity}}}{1-\delta_{3\sparsity}} \vectornorm{\noise}_2. \label{eq:15}
\end{align}
\end{lemma}

\begin{proof} %\subsection{Steps $ (ii) - (iii) $}
We know that $ \text{supp}(v_i) \in \mathcal{S}_i $. Starting from $ \vectornorm{v_i - \bestsignal}_2^2 $, the following holds true:
\begin{align}
\vectornorm{v_i - \bestsignal}_2^2 = \vectornorm{(v_i - \bestsignal)_{\mathcal{S}_i}}_2^2 + \vectornorm{(v_i - \bestsignal)_{\mathcal{S}_i^c}}_2^2. \label{eq:04}
\end{align}
Using the optimality condition, $ v_i $ is the minimizer of $ \vectornorm{\obs - \sensing v}_2^2 $ over the convex set $ \Theta = \lbrace v: \vectornorm{v}_1 \leq \lambda, \text{supp}(v) \in \mathcal{S}_i \rbrace $ and therefore:
\begin{align}
\langle \nabla f(v_i), (\bestsignal - v_i)_{\mathcal{S}_i} \rangle \geq 0 \Rightarrow \langle \sensing v_i - \obs, \sensing (v_i - \bestsignal)_{\mathcal{S}_i} \rangle \leq 0. \label{eq:05}
\end{align}
We calculate the following: 
\begin{align}
\vectornorm{(v_i - \bestsignal)_{\mathcal{S}_i}}_2^2 &= \langle v_i - \bestsignal, (v_i - \bestsignal)_{\mathcal{S}_i} \rangle \label{eq:06} \\
                            &\leq \langle v_i - \bestsignal, (v_i - \bestsignal)_{\mathcal{S}_i} \rangle - \langle \sensing v_i - \obs, \sensing (v_i - \bestsignal)_{\mathcal{S}_i} \rangle \label{eq:07} \\ 
                            &= \langle v_i - \bestsignal, (v_i - \bestsignal)_{\mathcal{S}_i} \rangle - \langle \sensing v_i - \sensing \bestsignal - \noise, \sensing (v_i - \bestsignal)_{\mathcal{S}_i} \rangle \label{eq:08} \\           
                            &= \langle v_i - \bestsignal, (v_i - \bestsignal)_{\mathcal{S}_i} \rangle - \langle v_i - \bestsignal, \sensing^\ast \sensing (v_i - \bestsignal)_{\mathcal{S}_i} \rangle + \langle \noise, \sensing (v_i - \bestsignal)_{\mathcal{S}_i} \rangle \label{eq:09} \\                            
                            &= \langle v_i - \bestsignal, (\mathbb{I} - \sensing^\ast \sensing)(v_i - \bestsignal)_{\mathcal{S}_i} \rangle + \langle \noise, \sensing (v_i - \bestsignal)_{\mathcal{S}_i} \rangle \label{eq:10} \\                            
                            &\leq | \langle v_i - \bestsignal, (\mathbb{I} - \sensing^\ast \sensing)(v_i - \bestsignal)_{\mathcal{S}_i} \rangle | + \langle \noise, \sensing (v_i - \bestsignal)_{\mathcal{S}_i} \rangle \label{eq:11} \\
                            &\stackrel{(i)}{\leq} \delta_{3\sparsity} \vectornorm{(v_i - \bestsignal)_{\mathcal{S}_i}}_2 \vectornorm{v_i - \bestsignal}_2 + \sqrt{1+\delta_{2\sparsity}} \vectornorm{(v_i - \bestsignal)_{\mathcal{S}_i}}_2 \vectornorm{\noise}_2,  \label{eq:12}                           
\end{align} where $ (i) $ comes from Cauchy-Swartz inequality and isometry constant definition. Simplifying the above quadratic expression, we obtain: 
\begin{align}
\vectornorm{(v_i - \bestsignal)_{\mathcal{S}_i}}_2 \leq \delta_{3\sparsity} \vectornorm{v_i - \bestsignal}_2 + \sqrt{1+\delta_{2\sparsity}} \vectornorm{\noise}_2. \label{eq:13}
\end{align}

As a consequence, (\ref{eq:04}) can be upper bounded by:
\begin{align}
\vectornorm{v_i - \bestsignal}_2^2 \leq (\delta_{3\sparsity} \vectornorm{v_i - \bestsignal}_2 + \sqrt{1+\delta_{2\sparsity}} \vectornorm{\noise}_2)^2+ \vectornorm{(v_i - \bestsignal)_{\mathcal{S}_i^c}}_2^2. \label{eq:14}
\end{align}

We form the quadratic polynomial for this inequality assuming as unknown variable the quantity $ \vectornorm{v_i - \bestsignal}_2 $. Bounding by the largest root of the resulting polynomial, we get:
\begin{align}
\vectornorm{v_i - \bestsignal}_2 \leq \frac{1}{\sqrt{1-\delta_{3\sparsity}^2}} \vectornorm{(v_i - \bestsignal)_{\mathcal{S}_i^c}}_2 + \frac{\sqrt{1+\delta_{2\sparsity}}}{1-\delta_{3\sparsity}} \vectornorm{\noise}_2. \label{eq:15a}
\end{align} 
\end{proof} 

\begin{lemma}{\label{lemma:comb_selection}}[Combinatorial selection] \textit{Let $ v_i $ be a $ 2\sparsity $-sparse proxy vector with indices in support set $ \mathcal{S}_i $, $ \constraint $ be a CSM and $ \gamma_i $ the projection of $ v_i $ under $ \constraint $. Then:}
\begin{align}
\vectornorm{\gamma_i  - v_i}_2^2 &\leq (1 - \epsilon) \vectornorm{(v_i - \bestsignal)_{\mathcal{S}_i}}_2^2 + \epsilon \vectornorm{v_i}_2^2. \label{eq:671}
\end{align}
\end{lemma}

\begin{proof}

Let $ \gamma_i^{\text{opt}} $ denote the optimal combinatorial projection of $ v_i $ under $ \constraint $, i.e.
\begin{align}
\gamma_i^{\text{opt}} = \mathcal{P}_{\constraint}(v_i) = \argmax_{(v_i)_{\mathcal{S}}: \mathcal{S} \in \mathcal{N}, \mathcal{S} \in \constraint} F(\mathcal{S}; v_i). \label{eq:65}
\end{align}

By the definition of the non-convex projection onto CSMs, it is apparent that:
\begin{align}
\vectornorm{\gamma_i^{\text{opt}} - v_i}_2 \leq \vectornorm{(v_i - \bestsignal)_{\mathcal{S}_i}}_2, \label{eq:19}
\end{align} over $ \constraint $ since $ \gamma_i^{\text{opt}} $ is the best approximation to $ v_i $ for that particular CSM.

In the general case, this step is performed approximately and we get $ \gamma_i $ as
\begin{align}
\gamma_i = \mathcal{P}_{\constraint}^{\epsilon}(v_i),
\end{align} an $ \epsilon $-approximate projection of $ v_i $ with corresponding variance reduction $ F(\widehat{\mathcal{S}}_\epsilon; v_i) $. % where $ \gamma_i = (v_i)_{\widehat{\mathcal{S}}_\epsilon} $. 
According to the definition of $ \text{PMAP}_{\epsilon} $, we calculate:
\begin{align}
F(\widehat{\mathcal{S}}_\epsilon;v_i) &\geq (1-\epsilon)\max_{\mathcal{S} \in \constraint} F(\mathcal{S};v_i) \Rightarrow \label{eq:66} \\
\vectornorm{v_i}_2^2 - \vectornorm{\gamma_i  - v_i}_2^2 &\geq (1 - \epsilon) \Big [\vectornorm{v_i}_2^2 - \vectornorm{\gamma_i^{\text{opt}} - v_i}_2^2 \Big ] \Rightarrow \label{eq:62} \\
\vectornorm{\gamma_i  - v_i}_2^2 &\leq (1 - \epsilon) \vectornorm{\gamma_i^{\text{opt}} - v_i}_2^2 + \epsilon \vectornorm{v_i}_2^2\label{eq:67} \Rightarrow \\
\vectornorm{\gamma_i  - v_i}_2^2 &\stackrel{(\ref{eq:19})}{\leq} (1 - \epsilon) \vectornorm{(v_i - \bestsignal)_{\mathcal{S}_i}}_2^2 + \epsilon \vectornorm{v_i}_2^2.
\end{align} 
\end{proof}

\begin{lemma}\label{lemma debias}[De-bias] \textit{Let $ \Gamma_i $ be the support set of a proxy vector $ \gamma_i $ where $ |\Gamma_i| \leq \sparsity $. Then, the least squares solution $ \signal_{i+1} $ given by:}
\begin{align}
\signal_{i+1} \leftarrow \argmin_{\signal: \vectornorm{\signal}_1 \leq \lambda, \text{supp}(\signal) \in \Gamma_i} \vectornorm{\obs - \sensing \signal}_2^2, \label{eq:580}
\end{align} \textit{satsifies:}
\begin{align}
\vectornorm{\signal_{i+1} - \bestsignal}_2 \leq \frac{1}{\sqrt{1 - \delta_{2\sparsity}^2}} \vectornorm{\gamma_i  - \bestsignal}_2 + \frac{\sqrt{1+\delta_\sparsity}}{1 - \delta_{2\sparsity}} \vectornorm{\noise}_2. \label{eq:58}
\end{align}
\end{lemma}

\begin{proof} 
The proof is similar to the proof of the Greedy descent step. Starting from $ \vectornorm{\signal_{i+1} - \bestsignal}_2^2 $:
\begin{align}
\vectornorm{\signal_{i+1} - \bestsignal}_2^2 = \vectornorm{(\signal_{i+1} - \bestsignal)_{\Gamma_i}}_2^2 + \vectornorm{(\signal_{i+1} - \bestsignal)_{\Gamma_i^c}}_2^2. \label{eq:200}
\end{align}
Similarly to lemma \ref{lemma:greedy}, $ \signal_{i+1} $ is the minimizer of $ \vectornorm{\obs - \sensing \signal}_2^2 $ under support set and norm constraints and therefore:
\begin{align}
\langle \nabla f(\signal_{i+1}), (\bestsignal - \signal_{i+1})_{\Gamma_i} \rangle \geq 0 \Rightarrow \langle \sensing \signal_{i+1} - \obs, \sensing (\signal_{i+1} - \bestsignal)_{\Gamma_i} \rangle \leq 0. \label{eq:201}
\end{align}
Following the same procedure, we have: 
\begin{align}
\vectornorm{(\signal_{i+1} - \bestsignal)_{\Gamma_i}}_2^2 &= \langle \signal_{i+1} - \bestsignal, (\signal_{i+1} - \bestsignal)_{\Gamma_i} \rangle \label{eq:202} \\
                            &\leq \langle \signal_{i+1} - \bestsignal, (\signal_{i+1} - \bestsignal)_{\Gamma_i} \rangle - \langle \sensing \signal_{i+1} - \obs, \sensing (\signal_{i+1} - \bestsignal)_{\Gamma_i} \rangle \label{eq:203} \\ 
                            &= \langle \signal_{i+1} - \bestsignal, (\signal_{i+1} - \bestsignal)_{\Gamma_i} \rangle - \langle \sensing \signal_{i+1} - \sensing \bestsignal - \noise, \sensing (\signal_{i+1} - \bestsignal)_{\Gamma_i} \rangle \label{eq:204} \\
                            &= \langle \signal_{i+1} - \bestsignal, (\signal_{i+1} - \bestsignal)_{\Gamma_i} \rangle - \langle \signal_{i+1} - \bestsignal, \sensing^\ast \sensing (\signal_{i+1} - \bestsignal)_{\Gamma_i} \rangle \nonumber \\ &+ \langle \noise, \sensing (\signal_{i+1} - \bestsignal)_{\Gamma_i} \rangle \label{eq:205}  \\
                            &= \langle \signal_{i+1} - \bestsignal, (\mathbb{I} - \sensing^\ast \sensing)(\signal_{i+1} - \bestsignal)_{\Gamma_i} \rangle + \langle \noise, \sensing (\signal_{i+1} - \bestsignal)_{\Gamma_i} \rangle \label{eq:206} \\                            
                            &\leq | \langle \signal_{i+1} - \bestsignal, (\mathbb{I} - \sensing^\ast \sensing)(\signal_{i+1} - \bestsignal)_{\Gamma_i} \rangle | + \langle \noise, \sensing (\signal_{i+1} - \bestsignal)_{\Gamma_i} \rangle \label{eq:207} \\
                            &\stackrel{(i)}{\leq} \delta_{2\sparsity} \vectornorm{(\signal_{i+1} - \bestsignal)_{\Gamma_i}}_2 \vectornorm{\signal_{i+1} - \bestsignal}_2 + \sqrt{1+\delta_{\sparsity}} \vectornorm{(\signal_{i+1} - \bestsignal)_{\Gamma_i}}_2 \vectornorm{\noise}_2, \label{eq:208}                           
\end{align} where $ (i) $ is due to Cauchy-Swartz inequality and isometry constant definition. Simplifying the above quadratic expression, we obtain 
\begin{align}
\vectornorm{(\signal_{i+1} - \bestsignal)_{\Gamma_i}}_2 \leq \delta_{2\sparsity} \vectornorm{\signal_{i+1} - \bestsignal}_2 + \sqrt{1+\delta_{\sparsity}} \vectornorm{\noise}_2. \label{eq:209}
\end{align}

Thus, $ \vectornorm{\signal_{i+1} - \bestsignal}_2^2 $ in eq. (\ref{eq:200}) can be upper bounded by the quadratic expression:
\begin{align}
\vectornorm{\signal_{i+1} - \bestsignal}_2^2 \leq (\delta_{2\sparsity} \vectornorm{\signal_{i+1} - \bestsignal}_2 + \sqrt{1+\delta_{\sparsity}} \vectornorm{\noise}_2)^2+ \vectornorm{(\signal_{i+1} - \bestsignal)_{\Gamma_i^c}}_2^2. \label{eq:210}
\end{align}

As in Lemma \ref{lemma:greedy}, we form a quadratic polynomial from (\ref{eq:210}) and bound $ \vectornorm{\signal_{i+1} - \bestsignal}_2 $ by the largest root. Thus, we obtain:
\begin{align}
\vectornorm{\signal_{i+1} - \bestsignal}_2 \leq \frac{1}{\sqrt{1-\delta_{2\sparsity}^2}} \vectornorm{(\signal_{i+1} - \bestsignal)_{\Gamma_i^c}}_2 + \frac{\sqrt{1+\delta_{\sparsity}}}{1-\delta_{2\sparsity}} \vectornorm{\noise}_2. \label{eq:211}
\end{align} In addition, we observe:
\begin{align}
\vectornorm{(\signal_{i+1} - \bestsignal)_{\Gamma_i^c}}_2 = \vectornorm{(\gamma_i - \bestsignal)_{\Gamma_i^c}}_2 \leq \vectornorm{\gamma_i - \bestsignal}_2,
\end{align} and thus:
\begin{align}
\vectornorm{\signal_{i+1} - \bestsignal}_2 \leq \frac{1}{\sqrt{1 - \delta_{2\sparsity}^2}} \vectornorm{\gamma_i  - \bestsignal}_2 + \frac{\sqrt{1+\delta_\sparsity}}{1 - \delta_{2\sparsity}} \vectornorm{\noise}_2. \label{eq:212}
\end{align}
\end{proof}

\begin{lemma} \textit{Let $ v_i $ be the least squares solution of the Greedy descent step given by}
\begin{align}
v_i \leftarrow \argmin_{v: \vectornorm{v}_1 \leq \lambda, \text{supp}(v) \in \mathcal{S}_i} \vectornorm{\obs - \sensing v}_2^2,
\end{align} \textit{and  $ \gamma_i $ be a proxy vector to $ v_i $ after applying Combinatorial selection and Least absolute shrinkage steps. % the projection of $ v_i $ over combinatorial constraints $ \constraint $ and $ l_1 $-norm constraints. 
Then, $ \vectornorm{\gamma_i  - \bestsignal}_2 $ can be expressed in terms of the distance from $ v_i $ to $ \bestsignal $ as follows:}
\begin{align}
\vectornorm{\gamma_i  - \bestsignal}_2 &\leq \sqrt{1 + \big((1-\epsilon) + 2\sqrt{1-\epsilon}\big)\delta_{3\sparsity}^2 + 2\delta_{3\sparsity}\sqrt{\epsilon} + \epsilon}\cdot \vectornorm{v_i - \bestsignal}_2 \nonumber \\ &+ D_1\vectornorm{\noise}_2 + D_2 \vectornorm{\bestsignal}_2 + D_3 \sqrt{\vectornorm{\bestsignal}_2 \vectornorm{\noise}_2}, \label{eq:74}
\end{align} \textit{where $ D_1, D_2, D_3 $ are constants depending on $ \epsilon, \delta_{2\sparsity}, \delta_{3\sparsity} $.}
\end{lemma}

\begin{proof} We observe the following
\begin{align}
\vectornorm{\gamma_i  - \bestsignal}_2^2 &= \vectornorm{\gamma_i  - v_i + v_i - \bestsignal}_2^2 \label{eq:16} \\
									   &= \vectornorm{(v_i - \bestsignal) - (v_i - \gamma_i )}_2^2 \label{eq:17} \\
									   &= \vectornorm{v_i - \bestsignal}_2^2 + \vectornorm{v_i - \gamma_i }_2^2 - 2\langle v_i - \bestsignal, v_i - \gamma_i  \rangle. \label{eq:18}
\end{align}
Focusing on the right hand side of expression (\ref{eq:18}), $ \langle v_i - \bestsignal, v_i - \gamma_i  \rangle = \langle v_i - \bestsignal, (v_i - \gamma_i )_{\mathcal{S}_i} \rangle $ can be similarly analysed as (\ref{eq:06})-(\ref{eq:12}) where we obtain the following expression:
\begin{align}
|\langle v_i - \bestsignal, (v_i - \gamma_i )_{\mathcal{S}_i} \rangle | \leq \delta_{3\sparsity} \vectornorm{v_i - \bestsignal}_2 \vectornorm{v_i - \gamma_i }_2 + \sqrt{1+\delta_{2\sparsity}} \vectornorm{v_i - \gamma_i }_2 \vectornorm{\noise}_2. \label{eq:24}
\end{align}

Now, expression (\ref{eq:18}) can be further transformed as:
\begin{align}
\vectornorm{\gamma_i  - \bestsignal}_2^2 &= \vectornorm{v_i - \bestsignal}_2^2 + \vectornorm{v_i - \gamma_i }_2^2 - 2\langle v_i - \bestsignal, v_i - \gamma_i  \rangle \label{eq:25} \\	
										&\leq \vectornorm{v_i - \bestsignal}_2^2 + \vectornorm{v_i - \gamma_i }_2^2 + 2|\langle v_i - \bestsignal, v_i - \gamma_i  \rangle |\label{eq:25a} \\	
										&\stackrel{(i)}{\leq} \vectornorm{v_i - \bestsignal}_2^2 + \vectornorm{v_i - \gamma_i }_2^2 +  2(\delta_{3\sparsity} \vectornorm{v_i - \bestsignal}_2 \vectornorm{v_i - \gamma_i }_2 + \sqrt{1+\delta_{2\sparsity}} \vectornorm{v_i - \gamma_i }_2 \vectornorm{\noise}_2) \label{eq:25b} \\
                                       &\stackrel{(ii)}{\leq} \vectornorm{v_i - \bestsignal}_2^2 + (1 - \epsilon) \vectornorm{\gamma_i^{\text{opt}} - v_i}_2^2 + \epsilon \vectornorm{v_i}_2^2 \nonumber \\ &+ 2\Big(\delta_{3\sparsity} \vectornorm{v_i - \bestsignal}_2 \sqrt{(1 - \epsilon) \vectornorm{\gamma_i^{\text{opt}} - v_i}_2^2 + \epsilon \vectornorm{v_i}_2^2} \nonumber \\ &+ \sqrt{1+\delta_{2\sparsity}} \sqrt{(1 - \epsilon) \vectornorm{\gamma_i^{\text{opt}} - v_i}_2^2 + \epsilon \vectornorm{v_i}_2^2} \vectornorm{\noise}_2\Big),\label{eq:68}
\end{align} where $ (i) $ is due to (\ref{eq:24}) and $ (ii) $ is due to Lemma \ref{lemma:comb_selection}. Given that $ \sqrt{a^2 + b^2} \leq a + b $ for $ a, b \geq 0 $, we further have:
\begin{align}
\vectornorm{\gamma_i  - \bestsignal}_2^2 &\leq \vectornorm{v_i - \bestsignal}_2^2 + (1 - \epsilon) \vectornorm{\gamma_i^{\text{opt}} - v_i}_2^2 + \epsilon \vectornorm{v_i}_2^2 + 2\delta_{3\sparsity} \vectornorm{v_i - \bestsignal}_2 \big(\sqrt{1 - \epsilon} \vectornorm{\gamma_i^{\text{opt}} - v_i}_2 + \sqrt{\epsilon} \vectornorm{v_i}_2\big) \nonumber \\ &+ 2\sqrt{1+\delta_{2\sparsity}} \big(\sqrt{1 - \epsilon} \vectornorm{\gamma_i^{\text{opt}} - v_i}_2 + \sqrt{\epsilon} \vectornorm{v_i}_2\big) \vectornorm{\noise}_2\label{eq:69} \\															   &\stackrel{(i)}{\leq} \vectornorm{v_i - \bestsignal}_2^2 + (1 - \epsilon) \vectornorm{(v_i - \bestsignal)_{\mathcal{S}_i}}_2^2 + \epsilon \vectornorm{v_i}_2^2 \nonumber \\ &+ 2\delta_{3\sparsity} \vectornorm{v_i - \bestsignal}_2 \big(\sqrt{1 - \epsilon} \vectornorm{(v_i - \bestsignal)_{\mathcal{S}_i}}_2 + \sqrt{\epsilon} \vectornorm{v_i}_2\big) \nonumber \\ &+ 2\sqrt{1+\delta_{2\sparsity}} \big(\sqrt{1 - \epsilon} \vectornorm{(v_i - \bestsignal)_{\mathcal{S}_i}}_2 + \sqrt{\epsilon} \vectornorm{v_i}_2\big) \vectornorm{\noise}_2\label{eq:70} \\
															   &\stackrel{(ii)}{\leq} \vectornorm{v_i - \bestsignal}_2^2 + (1 - \epsilon) (\delta_{3\sparsity} \vectornorm{v_i - \bestsignal}_2 + \sqrt{1+\delta_{2\sparsity}} \vectornorm{\noise}_2)^2 + \epsilon \vectornorm{v_i}_2^2 \nonumber \\ &+ 2\delta_{3\sparsity} \vectornorm{v_i - \bestsignal}_2 \big(\sqrt{1 - \epsilon} (\delta_{3\sparsity} \vectornorm{v_i - \bestsignal}_2 + \sqrt{1+\delta_{2\sparsity}} \vectornorm{\noise}_2) + \sqrt{\epsilon} \vectornorm{v_i}_2\big) \nonumber \\ &+ 2\sqrt{1+\delta_{2\sparsity}} \big(\sqrt{1 - \epsilon} (\delta_{3\sparsity} \vectornorm{v_i - \bestsignal}_2 + \sqrt{1+\delta_{2\sparsity}} \vectornorm{\noise}_2) + \sqrt{\epsilon} \vectornorm{v_i}_2\big) \vectornorm{\noise}_2,\label{eq:71}
\end{align}  where $ (i) $ is due to (\ref{eq:19}) and $ (ii) $ is due to (\ref{eq:13}).
Applying basic algebra on the right hand side of (\ref{eq:71}), we get:
\begin{align}
\vectornorm{\gamma_i  - \bestsignal}_2^2 &=\big(1 + (1-\epsilon)\delta_{3\sparsity}^2 + 2\delta_{3\sparsity}^2\sqrt{1-\epsilon}\big) \vectornorm{v_i - \bestsignal}_2^2 \nonumber \\ &+ \big(2(1-\epsilon)\delta_{3\sparsity}\sqrt{1+\delta_{2\sparsity}} + 4\delta_{3\sparsity}\sqrt{1-\epsilon}\sqrt{1+\delta_{2\sparsity}}\big)\vectornorm{v_i - \bestsignal}_2\vectornorm{\noise}_2 \nonumber \\ &+ \big((1-\epsilon)(1+\delta_{2\sparsity}) + 2(1+\delta_{2\sparsity})\sqrt{1-\epsilon}\big)\vectornorm{\noise}_2^2 \nonumber \\ &+ 2\delta_{3\sparsity}\sqrt{\epsilon}\vectornorm{v_i - \bestsignal}_2\vectornorm{v_i}_2 + 2\sqrt{\epsilon(1+\delta_{2\sparsity})}\vectornorm{v_i}_2\vectornorm{\noise}_2 + \epsilon \vectornorm{v_i}_2^2 \label{eq:72} \\
															   &\stackrel{(i)}{\leq} \Big(1 + \big((1-\epsilon) + 2\sqrt{1-\epsilon}\big)\delta_{3\sparsity}^2\Big)\Bigg(\vectornorm{v_i - \bestsignal}_2 + \sqrt{\frac{\big((1-\epsilon) + 2\sqrt{1-\epsilon}\big)(1+\delta_{2\sparsity})}{1 + \big((1-\epsilon) + 2\sqrt{1-\epsilon}\big)\delta_{3\sparsity}^2}}\vectornorm{\noise}\Bigg)^2 \nonumber \\
															   &+ 2\delta_{3\sparsity}\sqrt{\epsilon}\vectornorm{v_i - \bestsignal}_2\vectornorm{v_i}_2 + 2\sqrt{\epsilon(1+\delta_{2\sparsity})}\vectornorm{v_i}_2\vectornorm{\noise}_2 + \epsilon \vectornorm{v_i}_2^2. \label{eq:73}
\end{align} where $ (i) $ is obtained by completing the squares and eliminating negative terms in (\ref{eq:72}).

Using triangle inequality, we know that:
\begin{align}
\vectornorm{v_i}_2 \leq \vectornorm{v_i - \bestsignal}_2 + \vectornorm{\bestsignal}_2, \label{eq:100}
\end{align} and, thus, (\ref{eq:73}) can be further analyzed as:
\begin{align}
\vectornorm{\gamma_i - \bestsignal}_2^2 &\leq \Big(1 + \big((1-\epsilon) + 2\sqrt{1-\epsilon}\big)\delta_{3\sparsity}^2\Big)\Bigg(\vectornorm{v_i - \bestsignal}_2 + \sqrt{\frac{\big((1-\epsilon) + 2\sqrt{1-\epsilon}\big)(1+\delta_{2\sparsity})}{1 + \big((1-\epsilon) + 2\sqrt{1-\epsilon}\big)\delta_{3\sparsity}^2}}\vectornorm{\noise}\Bigg)^2 \nonumber \\
															   &+ (2\delta_{3\sparsity}\sqrt{\epsilon} + \epsilon)\vectornorm{v_i - \bestsignal}_2^2 + (2\delta_{3\sparsity}\sqrt{\epsilon}\vectornorm{\bestsignal}_2 + 2\sqrt{\epsilon(1+\delta_{2\sparsity})}\vectornorm{\noise}_2 + 2\epsilon \vectornorm{\bestsignal}_2) \vectornorm{v_i - \bestsignal}_2 \nonumber \\
															   &+ 2\sqrt{\epsilon(1+\delta_{2\sparsity})}\vectornorm{\bestsignal}_2\vectornorm{\noise}_2 + \epsilon \vectornorm{\bestsignal}_2^2. \label{eq:101}
\end{align}

After tedious computations, we end up with the following inequality:
\begin{align}
\vectornorm{\gamma_i  - \bestsignal}_2 &\leq \sqrt{1 + \big((1-\epsilon) + 2\sqrt{1-\epsilon}\big)\delta_{3\sparsity}^2 + 2\delta_{3\sparsity}\sqrt{\epsilon} + \epsilon}\cdot \vectornorm{v_i - \bestsignal}_2 \nonumber \\ &+ D_1\vectornorm{\noise}_2 + D_2 \vectornorm{\bestsignal}_2 + D_3 \sqrt{\vectornorm{\bestsignal}_2 \vectornorm{\noise}_2}, \label{eq:74a}
\end{align} where
\begin{align}
D_1 &\triangleq \frac{\sqrt{1 + \big((1-\epsilon) + 2\sqrt{1-\epsilon}\big)\delta_{3\sparsity}^2}\sqrt{\big((1-\epsilon) + 2\sqrt{1-\epsilon}\big)(1+\delta_{2\sparsity})} + \sqrt{\epsilon(1+\delta_{2\sparsity})}}{\sqrt{1 + \big((1-\epsilon) + 2\sqrt{1-\epsilon}\big)\delta_{3\sparsity}^2 + 2\delta_{3\sparsity}\sqrt{\epsilon} + \epsilon}}, \label{eq:102} \\
D_2 &\triangleq \frac{\delta_{3\sparsity}\sqrt{\epsilon} + \epsilon}{\sqrt{1 + \big((1-\epsilon) + 2\sqrt{1-\epsilon}\big)\delta_{3\sparsity}^2 + 2\delta_{3\sparsity}\sqrt{\epsilon} + \epsilon}} + \sqrt{\epsilon - \frac{(\epsilon + \delta_{3\sparsity}\sqrt{\epsilon})^2}{1 + \big((1-\epsilon) + 2\sqrt{1-\epsilon}\big)\delta_{3\sparsity}^2 + 2\delta_{3\sparsity}\sqrt{\epsilon} + \epsilon}}, \label{eq:103} \\
D_3 &\triangleq \sqrt{2\sqrt{\epsilon(1+\delta_{2\sparsity})}}. \label{eq:104}
\end{align} 
\end{proof}

Using the above lemmas, we now complete the proof of Theorem 1. 
\begin{proof} Combining (\ref{eq:15}) with (\ref{eq:74}), we get:
\begin{align}
\vectornorm{\gamma_i  - \bestsignal}_2 &\leq \sqrt{\frac{1 + \big((1-\epsilon) + 2\sqrt{1-\epsilon}\big)\delta_{3\sparsity}^2 + 2\delta_{3\sparsity}\sqrt{\epsilon} + \epsilon}{1-\delta_{3\sparsity}^2}}\cdot \vectornorm{(v_i - \bestsignal)_{\mathcal{S}_i^c}}_2 \nonumber \\ &+ D_4\vectornorm{\noise}_2 + D_2 \vectornorm{\bestsignal}_2 + D_3 \sqrt{\vectornorm{\bestsignal}_2 \vectornorm{\noise}_2}, \label{eq:75}
\end{align} where 
\begin{align}
D_4 \triangleq D_1 + \frac{\sqrt{1+\delta_{2\sparsity}}}{1-\delta_{3\sparsity}}\sqrt{1 + \big((1-\epsilon) + 2\sqrt{1-\epsilon}\big)\delta_{3\sparsity}^2 + 2\delta_{3\sparsity}\sqrt{\epsilon} + \epsilon}.  \label{eq:105}
\end{align}

We know that  $ \mathcal{X}_i \subseteq \mathcal{S}_i $. Thus, $ (v_i)_{\mathcal{S}_i^c} = 0 $ iff $ (\signal_i)_{\mathcal{S}_i^c} = 0 $. Therefore, 
\begin{align}
\vectornorm{(v_i - \bestsignal)_{\mathcal{S}_i^c}}_2 = \vectornorm{(v_i)_{\mathcal{S}_i^c} - (\bestsignal)_{\mathcal{S}_i^c}}_2 =  \vectornorm{(\signal_i)_{\mathcal{S}_i^c} - (\bestsignal)_{\mathcal{S}_i^c}}_2 = \vectornorm{(\signal_i - \bestsignal)_{\mathcal{S}_i^c}}_2. \label{eq:29}
\end{align}

Now, using (\ref{eq:99a}), we form the following recursion:
\begin{align}
\vectornorm{\gamma_i - \bestsignal}_2 &\leq \sqrt{\frac{1 + \big((1-\epsilon) + 2\sqrt{1-\epsilon}\big)\delta_{3\sparsity}^2 + 2\delta_{3\sparsity}\sqrt{\epsilon} + \epsilon}{1-\delta_{3\sparsity}^2}}(\delta_{3\sparsity} + \delta_{2\sparsity} + \sqrt{\epsilon}(1+\delta_{2\sparsity})) \vectornorm{\signal_i - \bestsignal}_2 \nonumber \\ &+ D_5\vectornorm{\noise}_2 + D_2 \vectornorm{\bestsignal}_2 + D_3 \sqrt{\vectornorm{\bestsignal}_2 \vectornorm{\noise}_2}, \label{eq:32}
\end{align} where 
\begin{align}
D_5 &\triangleq \sqrt{\frac{1 + \big((1-\epsilon) + 2\sqrt{1-\epsilon}\big)\delta_{3\sparsity}^2 + 2\delta_{3\sparsity}\sqrt{\epsilon} + \epsilon}{1-\delta_{3\sparsity}^2}}\big(\sqrt{2(1+\delta_{3\sparsity})} + \sqrt{\epsilon(1+\delta_{2\sparsity}})\big) + D_4. \label{eq:76}
\end{align}

Finally, substituting (\ref{eq:32}) in (\ref{eq:58}), we compute the desired recursive formula:
\begin{align}
\frac{\vectornorm{\signal_{i+1} - \bestsignal}_2}{\vectornorm{\bestsignal}_2} &\leq \rho \frac{\vectornorm{\signal_i - \bestsignal}_2}{\vectornorm{\bestsignal}_2} + \frac{c_1(\delta_{2\sparsity}, \delta_{3\sparsity}, \epsilon)}{SNR} + c_2(\delta_{2\sparsity}, \delta_{3\sparsity}, \epsilon) + c_3(\delta_{2\sparsity}, \delta_{3\sparsity}, \epsilon)\sqrt{\frac{1}{SNR}}, \label{eq:59a}
\end{align} where $ SNR = \frac{\vectornorm{\bestsignal}_2}{\vectornorm{\noise}_2} = \frac{\vectornorm{\bestsignal}_2}{\sqrt{f(\bestsignal)}} $ and
\begin{align}
\rho &\triangleq \frac{\delta_{3\sparsity} + \delta_{2\sparsity} + \sqrt{\epsilon}(1+\delta_{2\sparsity})}{\sqrt{1-\delta_{2\sparsity}^2}} \sqrt{\frac{1+\big((1-\epsilon) + 2\sqrt{1-\epsilon}\big)\delta_{3\sparsity}^2 + 2\delta_{3\sparsity}\sqrt{\epsilon} + \epsilon}{1-\delta_{3\sparsity}^2}}, \label{eq:79} \\
c_1(\delta_{2\sparsity}, \delta_{3\sparsity}, \epsilon) &\triangleq \frac{D_5}{\sqrt{1-\delta_{2\sparsity}^2}} + \frac{\sqrt{1+\delta_\sparsity}}{1-\delta_{2\sparsity}}, \label{eq:80} \\
c_2(\delta_{2\sparsity}, \delta_{3\sparsity}, \epsilon) &\triangleq \frac{1}{\sqrt{1-\delta_{2\sparsity}^2}}\Bigg(
\frac{\delta_{3\sparsity}\sqrt{\epsilon} + \epsilon}{\sqrt{1 + \big((1-\epsilon) + 2\sqrt{1-\epsilon}\big)\delta_{3\sparsity}^2 + 2\delta_{3\sparsity}\sqrt{\epsilon} + \epsilon}} \nonumber \\ &+ \sqrt{\epsilon - \frac{(\epsilon + \delta_{3\sparsity}\sqrt{\epsilon})^2}{1 + \big((1-\epsilon) + 2\sqrt{1-\epsilon}\big)\delta_{3\sparsity}^2 + 2\delta_{3\sparsity}\sqrt{\epsilon} + \epsilon}}\Bigg), \\
c_3(\delta_{2\sparsity}, \delta_{3\sparsity}, \epsilon) &\triangleq \frac{D_3}{\sqrt{1-\delta_{2\sparsity}^2}}.
\end{align}
\end{proof}

Some of the techniques used in the proof of Theorem 1 borrow from Foucart's paper \cite{foucart2010sparse}.

\bibliographystyle{unsrt}
\bibliography{recipes}

\end{document}